\documentclass[11pt,a4paper]{article}
\usepackage{iftex}

\usepackage{amsmath}
\usepackage{amssymb}
\usepackage{amsthm}
\usepackage{mathtools}
\usepackage{thmtools}
\usepackage[utf8]{inputenc}
\usepackage[T1]{fontenc}
\usepackage{isomath}
\usepackage{libertine}
\usepackage[libertine]{newtxmath} 
\usepackage[scaled=0.96]{zi4} 

\usepackage[DIV=11]{typearea} 

\usepackage{microtype}

\usepackage[algo2e,ruled,vlined,linesnumbered]{algorithm2e}

\usepackage{xcolor}

\usepackage[textwidth=2cm,textsize=footnotesize]{todonotes}
\usepackage{lineno}

\usepackage[
	style=alphabetic,
    maxalphanames=4,
    minalphanames=4,
    maxcitenames=4,
    minnames=1,
    maxbibnames=20,
	backref=true,
	doi=true,
	url=true,
	backend=biber
]{biblatex}

\usepackage[ocgcolorlinks]{hyperref} 

\newcommand*\samethanks[1][\value{footnote}]{\footnotemark[#1]}

\allowdisplaybreaks[1]

\makeatletter
\g@addto@macro\bfseries{\boldmath}
\makeatother

\makeatletter
\g@addto@macro\mdseries{\unboldmath}
\g@addto@macro\normalfont{\unboldmath}
\g@addto@macro\rmfamily{\unboldmath}
\g@addto@macro\upshape{\unboldmath}
\makeatother

\DeclareCiteCommand{\citem}
    {}
    {\mkbibbrackets{\bibhyperref{\usebibmacro{postnote}}}}
    {\multicitedelim}
    {}

\renewcommand*{\multicitedelim}{\addcomma\space}

\AtEveryCitekey{\clearfield{note}}

\makeatletter
\AtBeginDocument{%
    \newlength{\temp@x}%
    \newlength{\temp@y}%
    \newlength{\temp@w}%
    \newlength{\temp@h}%
    \def\my@coords#1#2#3#4{%
      \setlength{\temp@x}{#1}%
      \setlength{\temp@y}{#2}%
      \setlength{\temp@w}{#3}%
      \setlength{\temp@h}{#4}%
      \adjustlengths{}%
      \my@pdfliteral{\strip@pt\temp@x\space\strip@pt\temp@y\space\strip@pt\temp@w\space\strip@pt\temp@h\space re}}%
    \ifpdf
      \typeout{In PDF mode}%
      \def\my@pdfliteral#1{\pdfliteral page{#1}}
      \def\adjustlengths{}%
    \fi
    \ifxetex
      \def\my@pdfliteral #1{}
      \def\adjustlengths{\setlength{\temp@h}{-\temp@h}\addtolength{\temp@y}{1in}\addtolength{\temp@x}{-1in}}%
    \fi%
    \def\Hy@colorlink#1{%
      \begingroup
        \ifHy@ocgcolorlinks
          \def\Hy@ocgcolor{#1}%
          \my@pdfliteral{q}%
          \my@pdfliteral{7 Tr}
        \else
          \HyColor@UseColor#1%
        \fi
    }%
    \def\Hy@endcolorlink{%
      \ifHy@ocgcolorlinks%
        \my@pdfliteral{/OC/OCPrint BDC}%
        \my@coords{0pt}{0pt}{\pdfpagewidth}{\pdfpageheight}%
        \my@pdfliteral{F}
        \my@pdfliteral{EMC/OC/OCView BDC}%
        \begingroup%
          \expandafter\HyColor@UseColor\Hy@ocgcolor%
          \my@coords{0pt}{0pt}{\pdfpagewidth}{\pdfpageheight}%
          \my@pdfliteral{F}
        \endgroup%
        \my@pdfliteral{EMC}%
        \my@pdfliteral{0 Tr}
        \my@pdfliteral{Q}%
      \fi
      \endgroup
    }%
}
\makeatother

\colorlet{DarkRed}{red!50!black}
\colorlet{DarkGreen}{green!50!black}
\colorlet{DarkBlue}{blue!50!black}

\hypersetup{
	linkcolor = DarkRed,
	citecolor = DarkGreen,
	urlcolor = DarkBlue,
	bookmarksnumbered = true,
	linktocpage = true
}

\DontPrintSemicolon
\SetKwProg{Procedure}{Procedure}{}{}
\SetProcNameSty{textsc}
\SetFuncSty{textsc}
\SetKw{KwAnd}{and}
\SetKw{KwDo}{do}

\SetCommentSty{mycommfont}
\SetKwFunction{Insert}{insert}
\SetKwFunction{Delete}{delete}

\declaretheorem[numberwithin=section]{theorem}
\declaretheorem[numberlike=theorem]{lemma}

\declaretheorem[numberlike=theorem]{corollary}
\declaretheorem[numberlike=theorem]{definition}

\newcommand{\dist}{\operatorname{dist}}

\title{Bootstrapping Dynamic Distance Oracles}

\author{Sebastian Forster \thanks{Department of Computer Science, University of Salzburg, Salzburg, Austria. This work is supported by the Austrian Science Fund (FWF): P 32863-N. This project has received funding from the European Research Council (ERC) under the European Union's Horizon 2020 research and innovation programme (grant agreement No~947702).} \and Gramoz Goranci \thanks{Department of Computer Science, University of Vienna, Vienna, Austria} \and Yasamin Nazari\samethanks[1] \and Antonis Skarlatos\samethanks[1]}
\date{}

\hypersetup{
	pdftitle = {Bootstrapping Dynamic Distance Oracles},
	pdfauthor = {}
}

\begin{document}
\maketitle
\begin{abstract}
Designing approximate all-pairs distance oracles in the fully dynamic setting is one of the central problems in dynamic graph algorithms. Despite extensive research on this topic, the first result breaking the $O(\sqrt{n})$ barrier on the update time for any non-trivial approximation was introduced only recently by Forster, Goranci and Henzinger~[SODA'21] who achieved $m^{1/\rho+o(1)}$ amortized update time with a $O(\log n)^{3\rho-2}$ factor in the approximation ratio, for any parameter $\rho \geq 1$.

In this paper, we give the first \emph{constant-stretch} fully dynamic distance oracle with a small polynomial update and query time. Prior work required either at least a poly-logarithmic approximation or much larger update time. Our result gives a more fine-grained trade-off between stretch and update time, for instance we can achieve constant stretch of $O(\frac{1}{\rho^2})^{4/\rho}$ in amortized update time $\tilde{O}(n^{\rho})$, and query time $\tilde{O}(n^{\rho/8})$ for a constant parameter $\rho <1$. Our algorithm is randomized and assumes an oblivious adversary.



A core technical idea underlying our construction is to design a black-box reduction from decremental approximate hub-labeling schemes to fully dynamic distance oracles, which may be of independent interest. We then apply this reduction repeatedly to an existing decremental algorithm to bootstrap our fully dynamic solution.

\end{abstract}

\section{Introduction}


The All-Pairs Shortest Paths (APSP) problem is one of the cornerstone graph problems in combinatorial optimization. It has a wide range of applications, for instance in route planning, navigation systems, and routing in networks, and it has been extensively studied from both practical and theoretical perspectives. In theoretical computer science, this problem enjoys much popularity due to its historic contributions to the development of fundamental algorithmic tools and definitions as well as being used as a subroutine for solving other problems.

The APSP problem has also been studied extensively in \emph{dynamic} settings. Here, the underlying graph undergoes edge insertions and deletions (referred to as edge \emph{updates}), and the goal is to quickly report an approximation to the shortest paths between \emph{any} source-target vertex pair. The dynamic setting is perhaps even more realistic for some of the applications of the APSP problem, e.g., in navigation systems, link statistics of road networks are prone to changes because of evolving traffic conditions. A naive (but rather expensive) solution to handle the updates is achieved by running an exact static algorithm after each update. However, at an intuitive level, one would expect to somehow exploit the fact that a single update is small compared to the size of the network, and thus come up with much faster update times.

Much of the research literature in dynamic APSP has focused on the \emph{partially} dynamic setting. In contrast to the \emph{fully} dynamic counterpart, this weaker model restricts the types of updates to edge insertions or deletions only. Some reasons for studying partially dynamic algorithms include their application as a subroutine in speeding up static algorithms (e.g., flow problems~\cite{Madry10}), or their utilization as a stepping stone for designing fully-dynamic algorithms, something that we will also exploit in this work.
The popularity of the partially dynamic setting can also attributed to the fact that dealing with only one type of update usually leads to better algorithmic guarantees. In fact, the fully dynamic APSP problem admits strong conditional lower bounds in the \textit{low approximation} regimes: under plausible hardness assumptions, Abboud and Vassilevska Williams \cite{AbboudW14}, and later Henzinger, Krinninger, Nanongkai, and Saranurak \cite{HenzingerKNS15} show that there are no dynamic APSP algorithms achieving a $(3-\epsilon)$ approximation with sublinear query time and the update time being a small polynomial.

From an upper bounds perspective, there are only two works that achieve sublinear update time for fully dynamic APSP. Abraham, Chechik, and Talwar \cite{AbrahamCT14} showed that there is an algorithm that achieves constant approximation and sublinear update time. However, their algorithm cannot break the $O(\sqrt{n})$ barrier on the update time. Forster, Goranci, and Henzinger \cite{ForsterGH21} gave different trade-offs between approximation and update time. In particular, in $n^{o(1)}$ amortized update time and polylogarithmic query time they achieve $n^{o(1)}$ approximation. These two works suffer from either a large approximation guarantee or update time, leaving open the following key question:
\begin{center}
    \emph{Is there a fully dynamic APSP algorithm that achieves \emph{constant} approximation with a very small polynomial update time?}
\end{center}

\subsection{Our result}

In this paper, we answer the question of achieving constant approximation with a very small polynomial update time for the fully dynamic APSP in the affirmative, also known as the \emph{fully dynamic distance oracle} problem. More generally, we obtain a trade-off between approximation, update time, and query time as follows:

\begin{restatable}{corollary}{maincor}\label{cor:main_oracle}
Given a weighted undirected graph $G=(V,E)$ and a constant parameter $0 < \rho < 1$, there is a randomized, fully dynamic distance oracle with constant stretch $(\frac{256}{\rho^2})^{4/\rho}$ that w.h.p. achieves $\tilde{O}(n^\rho)$ amortized update time and $\tilde{O}(n^{\rho/8})$ query time. These guarantees work against an oblivious adversary.
\end{restatable}

In addition to the constant stretch regime, we obtain several interesting tradeoffs, as shown in Theorem \ref{thm:general_tradeoffs}. For example, our algorithm achieves $O(\log \log n)$ stretch with a much faster query time of $n^{o(1)}$ and very small polynomial update time (see Corollary \ref{cor:loglog_stretch}).

Our result brings the algorithmic guarantees on fully dynamic distance oracles closer to the recent conditional hardness result by Abboud, Bringmann, Khoury, and Zamir \cite{AbboudBKZ22} (and the subsequent refinement in \cite{AbboudBF22}), who showed that there is no fully dynamic algorithm that simultaneously achieves constant approximation and $n^{o(1)}$ update and query time. We also remark that our results are consistent with their lower bound since if we insist on constant approximation, the above trade-off shows that the update time cannot be made as efficient as $n^{o(1)}$.

On the technical side, our result follows the widespread ``high-level'' approach of extending decremental algorithms to the fully dynamic setting (see e.g.~\cite{henzinger1995fully,RodittyZ08,RodittyZ11,RZ12,Bernstein16,HenzingerKN16,AbrahamCT14,ForsterGH21}) and it is inspired by recent developments on the dynamic distance oracle literature that rely on vertex sparsification~\cite{ForsterGH21,ChenGHPS20,FNG22}. Specifically, we design a reduction that turns a decremental hub-labeling scheme with
some specific properties
into a fully dynamic distance oracle, which may be of independent interest. Our key observation is that an existing state-of-the-art decremental distance oracle that works against an oblivious adversary can serve as such hub-labeling scheme. The fully dynamic distance oracle is then obtained by repeatedly applying the reduction whilst carefully tuning various parameters across levels in the hierarchy. 

More generally, our reduction does not make any assumptions on the adversary and is based on properties that are quite natural. At a high-level, we consider decremental approximate hub labeling scheme with the following properties. (1) For every vertex $v \in V$, maintain a set $S(v)$, called a \textit{hub set}, that has bounded size.
(2) For every vertex $v \in V$, maintain distance estimates $\delta(v, u)$ for each $u \in S(v)$, with bounded \emph{recourse}, which is defined as the number of times such distance estimates are affected during the execution of the algorithm.
(3) Return the final estimate between a pair of vertices $s,t \in V$, by minimizing estimates over  elements in $S(s) \cap S(t)$.

Many known distance oracles (e.g.~variants of the well-known distance oracle of \cite{TZ2005}) have a query mechanism that satisfies the first and third properties, while efficient dynamic distance oracles are often based on bounded recourse structures satisfying the second property.

Hence we hope that this reduction can be further utilized in the future by characterizing deterministic decremental distance oracles or the ones with different stretch/time tradeoffs as such hub-labeling schemes. Similar reductions have been previously proposed in \cite{AbrahamCT14} and then refined in \cite{ForsterGH21} in slightly different contexts. In this work, in addition to refining this approach for obtaining a constant stretch distance oracle, we aim to keep the reduction as modular as possible to facilitate potential future applications.

\subsection{Related Work}
In the following, we give an overview of existing works on fully dynamic all-pairs distance oracles by dividing them into several categories based on their stretch guarantee.
Unless noted otherwise, all algorithms cited in the following are randomized and have amortized update time.
We report running time bounds for constant accuracy parameter $ \epsilon $ and assume that we are dealing with graphs with positive integer edge weights that are polynomial in the number of vertices.
We would also like to point out that all ``combinatorial'' algorithms discussed in the following (i.e., algorithms that do not rely on ``algebraic'' techniques like dynamic matrix inverse) are internally employing decremental algorithms.
Decremental algorithms have also been studied on their own with various tradeoffs~\cite{RZ12,Bernstein16,henzinger2014decremental,Chechik18,LackiN22,DoryFNV2022}, and competitive deterministic algorithms have been devised, e.g., ~\cite{HenzingerKN16,BernsteinGS21,Chuzhoy21}. 

\paragraph{Exact.}
After earlier attempts on the problem~\cite{King99,DemetrescuI06}, Demetrescu and Italiano~\cite{DemetrescuI04} presented their seminal work on exact distance maintenance achieving $ \tilde O (n^2) $ update time (with log-factor improvements by Thorup~\cite{Thorup04}) and constant query time for weighted directed graphs.

Subsequently, researchers have developed algorithms with subcubic worst-case update time and constant query time~\cite{Thorup05,AbrahamCK17} with some of them being deterministic~\cite{GutenbergW20b, ChechikZ23}.
Note that one can construct a simple update sequence for which any fully dynamic algorithm maintaining the distance matrix or the shortest path matrix explicitly needs to perform $ \Omega (n^2) $ changes to this matrix per update.

Algorithms breaking the $n^2$ barrier at the cost of large query time have been obtained in unweighted directed graphs by Roditty and Zwick~\cite{RodittyZ11} (update time $ \tilde O (m n^2 / t^2) $ and query time $ O (t) $ for any $ \sqrt{n} \leq t \leq n^{3/4} $), Sankowski~\cite{Sankowski05} (worst-case update time $ O (n^{1.897}) $ and query time $ O (n^{1.265}) $), and van den Brand, Nanongkai, and Saranurak~\cite{BrandNS19} (worst-case update time $ O (n^{1.724}) $ and query time $ O (n^{1.724}) $).
The latter two approaches are algebraic and their running time bounds depend on the matrix multiplication coefficient $ \omega $.

\paragraph{$(1+\epsilon)$-approximation.} In addition to exact algorithms, combinatorial and algebraic algorithms have also been developed for the low stretch regime of $(1+\epsilon)$-approximation.
In particular, Roditty and Zwick obtained the following trade-off with a combinatorial algorithm: update time $\tilde{O}(mn/t)$ and query time of $\tilde{O}(t)$ for any $ \delta > 0$ and $ t \leq m^{1/2 - \delta}$.
Subsequently, for $ t \leq \sqrt{n} $, a deterministic variant was developed~\cite{HenzingerKN16} and it was generalized to weighted, directed graphs~\cite{Bernstein16}.
Furthermore, by a standard reduction (see e.g.~\cite{BernsteinPW19}) using a decremental approximate single-source shortest paths algorithm~\cite{henzinger2014decremental,BernsteinGS21}, one obtains a combinatorial, deterministic algorithm with update time $ O (n m^{1+o(1)} / t) $ and query time $ t $ for any $ t \leq n $, for the fully dynamic all-pairs problem in weighted undirected graphs.
Conditional lower bounds~\cite{Patrascu10,AbboudW14,HenzingerKNS15} suggest that the update and the query time cannot be both small polynomials in $ n $.
For example, no algorithm can maintain a $ (5/3 - \epsilon) $-approximation with update time $ O (m^{1/2-\delta}) $ and query time $ O (m^{1-\delta}) $ for any $ \delta > 0 $, unless the OMv conjecture fails~\cite{HenzingerKNS15}.

Algebraic approaches can achieve subquadratic update time and sublinear query time, namely worst-case update time $ O(n^{1.863}) $ and query time $ O (n^{0.666}) $ in weighted  directed graphs~\cite{BrandN19}, or worst-case update time $ O(n^{1.788}) $ and query time $ O (n^{0.45}) $ in unweighted  undirected graphs~\cite{BFN22}.
As the conditional lower bound by Abboud and Vassilevska Williams~\cite{AbboudW14} shows, algebraic approaches seem to be necessary in this regime: 
unless one is able to multiply two $n\times n$ Boolean matrices in $ O(n^{3 - \delta})$ time for some constant $\delta>0$, no fully dynamic algorithm for $st$ reachability in directed graphs can have $ O(n^{2-\delta'}) $ update and query time and $ O(n^{3 - \delta'}) $ preprocessing time (for some constant $ \delta'> 0 $).

\paragraph{$(2+\epsilon)$-approximation.} Apart from earlier work~\cite{King99}, the only relevant algorithm in the $(2+\epsilon)$-approximation regime is by Bernstein~\cite{Bernstein09} and achieves update time $ m^{1 + o(1)} $ and query time $ O (\log \log \log n) $ in weighted undirected graphs.
It can be made deterministic using the deterministic approximate single-source shortest path algorithm by Bernstein, Probst Gutenberg, and Saranurak~\cite{BernsteinGS21}.
The only conditional lower bound in this regime that we are aware of states that no algorithm can maintain a $ (3 - \epsilon) $-approximation with update time $ O (n^{1/2-\delta}) $ and query time $ O (n^{1-\delta}) $ for any $ \delta > 0 $, unless the OMv conjecture fails~\cite{HenzingerKNS15}.

\paragraph{Larger approximation.}
In the regime of stretch at least~$ 3 $, the following trade-offs between stretch and update time have been developed:
Abraham, Chechik, and Talwar~\cite{AbrahamCT14} designed an algorithm for unweighted  undirected graphs with stretch $ 2^{O (\rho k)} $, update time $ \tilde O (m^{1/2} n^{1/k}) $, and query time $ O (k^2 \rho^2) $, where $ k \geq 1 $ is a freely chosen parameter and $ \rho = 1 + \lceil \log n^{1-1/k} / \log (m / n^{1-1/k})  \rceil $.
Forster, Goranci, and Henzinger~\cite{ForsterGH21} designed an algorithm for weighted  undirected graphs with stretch $ O (\log n)^{3 k - 2} $, update time time $ m^{1/k + o(1)} \cdot O(\log n)^{4k - 2} $, and query time $ O (k (\log n)^2) $, where $ k \geq 2 $ is an arbitrary integer parameter.
Finally, note that any algorithm whose update time depends on the sparsity of the graph (possibly also a static one) can be run on a spanner of the input graph maintained by a fully dynamic spanner algorithm~\cite{BaswanaKS12}.
These upper bounds are complemented by the following conditional lower bound: for any integer constant $ k \geq 2 $, there is no dynamic approximate distance oracle with stretch $ 2k-1 $, update time $ O (m^u) $ and query time $ O(m^q) $ with $ ku + (k+1)q < 1 $, unless the 3-SUM conjecture fails~\cite{AbboudBF22}.


\section{Preliminaries}
We consider weighted undirected graphs $G = (V, E, w)$ with positive integer edge weights. We denote by $n = |V|$ the number of vertices, by $m = |E|$ the number of edges, and by $W$ the maximum weight of an edge.
We denote by $\dist_G(u, v)$ the value of a shortest path from $u$ to $v$ in $G$.

In dynamic graph algorithms, the graph is subject to updates and the algorithm has to process these updates by spending as little time as possible. 
In this paper, we consider updates
that insert a single edge to the graph or delete a single edge from the graph. Moreover, observe that an update that changes the weight of an edge can be simulated by two updates, where the first update deletes
the corresponding edge and the second update re-inserts the edge with the new weight. Let $G^{(0)}$ be 
the initial graph, and $G^{(t)}$ be the graph at time $t$ which is the time after $t$ updates have been performed to the graph.

In this paper we are interested in designing \emph{fully dynamic} algorithms which can process
edge insertions and edge deletions, and thus, weight changes as well. A \emph{decremental} algorithm can process only edge deletions and weight increases. We assume that the updates to the graph are performed by an \emph{oblivious adversary} who fixes the sequence of updates before the algorithm starts. Namely,
the adversary cannot adapt the updates based on the choices of the algorithm during the execution. 
We say that an algorithm has \emph{amortized update time} $u(n, m)$ if its total time spent for processing any sequence of $\ell$ updates is bounded by $\ell \cdot u(n, m)$, when
it starts from an empty graph with $n$ vertices and during all the updates has at most $m$ edges (the time needed to initialize the algorithm on the empty graph before the first update is also included).

In our analysis we use $\tilde{O}(1)$ to hide factors polylogarithmic in $nW$. 
Namely, we write $\tilde{O}(1)^d$ to represent the term $O(\log^{cd} nW)$, for a constant $c$ and a parameter $d$.

\section{Fully Dynamic Distance Oracle}

The technical details of our distance oracle are divided into three parts. Initially in Section \ref{subsec:reduction}, we give the definition of a hub-labeling scheme together with other useful definitions. Afterwards,
we provide a reduction for extending a decremental approximate hub-labeling scheme properties to a fully dynamic distance oracle. Then in Section \ref{subsec:decremental}, we explain how an existing decremental algorithm gives us an approximate hub-labeling scheme that we can use in this reduction, and finally in Section \ref{subse:final_result} we put everything together by applying
our reduction repeatedly, in order to get a family of fully dynamic distance oracles.

\subsection{Reduction from a decremental hub-labeling scheme to fully dynamic distance oracle} \label{subsec:reduction}

We start by defining approximate hub-labeling schemes, and then explain how they are used in our reduction. Hub-labeling schemes were formally defined by \cite{abraham2012} (and were previously introduced under the name $2$-hop cover\footnote{The concept of $2$-hop cover or hub labeling should not be confused with the (related) concept of a hopset that we will later see in Section \ref{subsec:decremental}.} in \cite{cohen2003}). We are slightly modifying the definition for our purpose, for instance by considering an approximate variant.

\begin{definition}[Approximate Hub-Labeling Scheme]
\label{def:apprx_hub_lab_schm}
Given a graph $G = (V, E)$, a \emph{hub-labeling scheme~$\mathcal{L}$} of \emph{stretch~$\alpha$} consists of 

\begin{enumerate}
    \item for every vertex $v \in V$, a \emph{hub set} $S(v) \subseteq V$ and 
    \item for every pair of vertices $u, v \in V$, a distance estimate $\delta(v, u)$ such that
    $\dist_G(v, u) \leq \delta(v, u) < \infty$ if $u \in S(v)$ and $\delta(v, u) = \infty$ otherwise.
\end{enumerate}
and for every pair of vertices $ s $ and $ t $ guarantees that
\begin{equation*}
    \delta_{\mathcal{L}} (s, t) := \min_{v \in S(s) \cap S(t)} (\delta (s, v) + \delta (t, v)) \leq \alpha \cdot \dist_G (s, t)  \, .
\end{equation*}
\end{definition}

The \textit{distance label} of a vertex $v$ consists of the hub set $S(v)$ and the corresponding distance estimates $\delta(v,u)$, for all $u \in S(v)$.

Note that the definition implies $ \delta_{\mathcal{L}} (s, t) \geq \dist_G (s, t) $ for every pair of vertices $ s $ and $ t $.
Furthermore, a hub-labeling scheme of stretch $ \alpha $ directly implements a distance oracle of stretch $ \alpha $ with query time $ O (\max_{v \in V} |S(v)|) $ that consists of the collection of distance labels for all vertices $v \in V$.
We also remark that the entries of value $ \infty $ in the distance estimate $ \delta (\cdot, \cdot) $ do not need to be stored explicitly if the hub sets are stored explicitly and that the distance estimate $\delta(\cdot, \cdot)$ is not necessarily symmetric.

In the following we consider \emph{decremental} algorithms for maintaining approximate hub-labeling schemes, that is, \emph{decremental approximate hub-labeling schemes} which process each edge deletion in the graph by first updating their internal data structures and then outputting the changes made to the hub sets and the distance estimates $ \delta (\cdot, \cdot) $.
Namely for a vertex $v \in V$, vertices may leave or join $S(v)$, or the distance estimates of vertices belonging to $S(v)$ may change, since the decremental algorithm has to update this information for maintaining correctness at query time.

Denote by $S^{(t)}(v)$ the hub set of a vertex $v \in V$, after $t$ updates have been processed by the decremental approximate hub-labeling scheme (we may omit the superscript $t$ whenever time is fixed), where $t \geq 1$ is an integer parameter. 
Then for a pair of vertices $u, v \in V$, the distance estimate $\delta(v, u)$ after $t$ updates is defined based on Definition~\ref{def:apprx_hub_lab_schm} and $S^{(t)}(v)$. Namely, if 
$u$ is inside the hub set of $v$ after $t$ updates (i.e., $u \in S^{(t)}(v)$) then $\dist_{G^{(t)}}(v, u) \leq \delta(v, u) < \infty$, otherwise $\delta(v, u) = \infty$. 

After $t$ edge deletions processed by the decremental approximate hub-labeling scheme, there are three possible changes to the distance estimates $\delta(v, \cdot)$ corresponding to a vertex $v \in V$.
(1) The distance estimate $\delta(v, u)$ changes for a vertex $u \in S^{(t-1)}(v) \cap S^{(t)}(v)$ that remains inside the hub set of $v$. 
(2) The distance estimate $\delta(v, u)$ becomes $\infty$ because a vertex $u \in S^{(t-1)}(v) \setminus S^{(t)}(v)$ leaves the hub set of $v$.
(3) The distance estimate $\delta(v, u)$ receives a finite value because a vertex $u \in S^{(t)}(v) \setminus S^{(t-1)}(v)$ enters the hub set of $v$. 
Let $\chi^{(t)}(v)$ be the number of all these changes to $\delta(v, \cdot)$ corresponding to $v$ at time $t$, and $X(v) = \sum_{t} \chi^{(t)}(v)$ be
the total number of such changes to $\delta(v, \cdot)$ corresponding to $v$ over the course of the algorithm.

In the following lemma, we present a reduction from a decremental approximate hub-labeling scheme to a fully dynamic distance oracle.

\begin{lemma}\label{lem:decr_to_fully_dynamic_dist_oracle}
Consider a decremental hub-labeling scheme $\mathcal{A}$ of stretch $\alpha$ with total update time $T_\mathcal{A} (n, m, W)$ and query time $ Q_\mathcal{A}(n, m, W)$,
with the following properties:
\begin{enumerate}
    \item $\forall v \in V$ and $ \forall t: |S^{(t)}(v)| \leq \gamma$. In other words, the size of the hub set of any vertex is bounded by $\gamma$ at any moment of the algorithm.
    \item $\forall v \in V: X(v) \leq \zeta$. In other words, for every vertex $v \in V$ the total number of changes to $\delta(v, \cdot)$ is at most $\zeta$ over the course of the algorithm. Moreover the algorithm detects and reports these changes explicitly.
\end{enumerate}
Then given $\mathcal{A}$ and a fully dynamic distance oracle $\mathcal{B}$ of stretch $\beta$ with amortized update time $t_\mathcal{B}(n, m, W)$ and query time $Q_\mathcal{B} (n, m, W)$, for any integer $\ell \geq 1$, there is a fully dynamic distance oracle $\mathcal{C}$ of stretch $\alpha \beta$ with amortized update time $t_\mathcal{C}(n, m, W) = T_\mathcal{A}(n, m, W) / \ell + t_\mathcal{B}(\min(\ell(2 + 2\mu), n), \ell(1 + 2\mu), nW) \cdot (2 + 4\mu) $ and query time $Q_\mathcal{C}(n, m, W) = Q_\mathcal{A}(n, m, W) + \gamma^2 \cdot Q_\mathcal{B}(\min\{\ell(2 + 2\mu), n\}, \ell(1 + 2\mu), nW)$, where $\mu = \gamma + \zeta$.
\end{lemma}
\begin{proof}
We organize the proof in three parts. The first part gives the reduction from $\mathcal{A}$ and $\mathcal{B}$ to $\mathcal{C}$, and the second and third part concerns the correctness and the running times respectively.
\paragraph{Reduction.}
The fully dynamic distance oracle $\mathcal{C}$ proceeds in phases of length $\ell$.
At the beginning of the first phase (which is also the beginning of the algorithm), $\mathcal{C}$ initializes the fully dynamic distance oracle $\mathcal{B}$ on the initially empty graph $G$ on $2\ell$ vertices\footnote{\label{footnote:no_vert_ins}This minor technical detail makes sure that $\mathcal{B}$ does not have to deal with vertex insertions.} and sets an update counter to $0$. Whenever an update to $G$ occurs in the first phase, the update is directly processed by $\mathcal{B}$.\footnote{The special treatment of the first $\ell$ updates is just a technical necessity for a rigorous amortization argument in the running time analysis.}
As soon as the number of updates is more than $\ell$, the second phase is started. 
We define several sets and the graph $H$ that the fully dynamic distance oracle $\mathcal{C}$ maintains during each subsequent phase:
\begin{itemize}
    \item Let $F$ be the set of edges present in $G$ at the beginning of the phase, $E$ be the current set of edges in $G$,
    and $D$ be the set of edges deleted from $G$ during the phase.
    \item Let $I = E \setminus (F \setminus D)$ be the set of edges inserted to $G$ since the beginning of the phase without subsequently having been deleted during the phase, and $U = \{v \in V \mid \exists e \in I: v \in e\}$ be the set of endpoints of edges in $I$.
    \item Let $H$ be the auxiliary graph that consists of all edges $(u, v) \in I$, together with their hub sets $S^{(t)}(u)$
    and $S^{(t)}(v)$ after $t$ edge deletions have been processed by $\mathcal{A}$. Specifically, $V(H) = \{v \in V \mid v \in U \text{ or } (u \in U \text{ and } v \in S^{(t)}(u))\}$
    and $E(H) = \{(u, v) \mid (u, v) \in I \text{ or } (v \in U \text{ and } u \in S^{(t)}(v))\}$. Note that at any fixed moment, the size of $V(H)$ is at most $\ell \cdot (2 + 2\gamma)$ and the size of $E(H)$ is at most $\ell \cdot (1 + 2\gamma)$.
\end{itemize}

At the beginning of each subsequent phase, $\mathcal{C}$ stores $F, E$ and $H$, and sets an update counter to $0$. Furthermore, $\mathcal{C}$ initializes the decremental approximate hub-labeling scheme $\mathcal{A}$ on the current graph $G$, and the fully dynamic distance oracle $\mathcal{B}$ on $H$ which is initially an empty ``sketch'' graph on $\ell \cdot (2 + 2 \mu)$ vertices.
The graph $H$ can be thought of as responsible for maintaining estimates for paths that use
inserted edges.

Whenever an update to $G$ occurs, $\mathcal{C}$ first checks via the update counter whether the number of updates since the beginning of the phase is more than $\ell$.
If this is the case, then $\mathcal{C}$ starts a new phase. Otherwise, after an update the fully dynamic distance oracle $\mathcal{C}$ does the following.
On the insertion of an edge $(u, v)$ to $G$, $\mathcal{C}$ adds $(u, v)$ to $I$, adds
$u$ and $v$ to $U$, and adds the edge $(u, v)$ to $H$, together with the edges $(u, p)$ for every $p \in S(u)$ and $(v, p)$ for every $p \in S(v)$.
Any time an edge $ (u, v) $ is added to $ H $, its weight is set to:
\begin{equation*}
w_H (u, v) = \min (w_G (u, v), \delta (u, v), \delta(v, u)).
\end{equation*}
Whenever the first edge incident to some vertex $ v $ is added to $ H $, the algorithm finds a ``fresh'' vertex (of degree $ 0 $) in $ H $ and henceforth identifies it as $ v $.
This is always possible, since by the two properties,
the number of such vertices in a phase of length $ \ell $ is at most $\ell \cdot (2 + 2\mu)$.
On the deletion of an edge $(u, v) \in E$ from $G$, there are two cases to consider. 
\begin{enumerate}
    \item 
If the edge $(u, v)$ was not present at the beginning of the current phase, or has been deleted and re-inserted (i.e., $(u, v) \in I$), $\mathcal{C}$
removes $(u, v)$ from $I$, adds $(u, v)$ to $D$, and updates the set $U$ and the graph $H$ accordingly.
In particular, if $u \in U$ and $v \in S(u)$, or $v \in U$ and $u \in S(v)$, $\mathcal{C}$ updates the weight of the edge $(u, v)$ in $H$ to
$w_H (u, v) = \min(\delta (u, v), \delta(v, u))$ (as $w_G (u, v) = \infty$ after the deletion),
otherwise $\mathcal{C}$ removes $(u, v)$ from $H$. Also, for all the vertices $v$ that left $U$ and all the edges $(v, p) \in E(H)$ such that $p \in S(v)$, if $p \in U$ and $v \in S(p)$, then $\mathcal{C}$ updates the weight of $(v, p)$ in $H$ to $w_H (v, p) = \delta (p, v)$ (as $v \notin U$ after the deletion), and otherwise $\mathcal{C}$ removes $(v, p)$ from $H$.
\item If the edge $(u, v)$ was present at the beginning of the current phase and has not been deleted yet (i.e., $ (u, v) \in F \setminus D$), $\mathcal{C}$
adds $(u, v)$ to $D$ and the deletion is processed by $\mathcal{A}$. Whenever $\mathcal{A}$ changes a distance estimate $\delta(v, \cdot)$ corresponding to a vertex $v \in V$ and its hub set $S(v)$,
$\mathcal{C}$ updates the graph $H$ accordingly. In particular, there are three possible scenarios at time $t$ of $\mathcal{A}$.\footnote{Note that $t$ is the number of updates processed only by $\mathcal{A}$ during the phase.}
(1) Whenever the value of $\delta (v, u)$ changes for a vertex $u \in S^{(t-1)}(v)
\cap S^{(t)}(v)$ that remains inside the hub set of $v$, $\mathcal{C}$ updates the weight of the edge $(v, u)$ in $H$ to
$w_H (v, u) = \min (w_G (v, u), \delta (v, u), \delta(u, v))$.
(2) Whenever a vertex $u \in S^{(t-1)}(v) \setminus S^{(t)}(v)$ leaves the hub set of $v$,
then if $(v, u) \in I$ or $u \in U$ and $v \in S(u)$,
$\mathcal{C}$ updates the weight of the edge $(v, u)$ in $H$ to $w_H(v, u) = \min(w_G(v, u),
\delta(u, v))$ (as $\delta (v, u) = \infty$ after the deletion), otherwise $\mathcal{C}$ removes $(v, u)$ from $H$.
(3) Whenever a vertex $u \in S^{(t)}(v) \setminus S^{(t-1)}(v)$ enters the hub set of $v$, then if $(v, u) \in I$ or $u \in U$ and $v \in S(u)$, $\mathcal{C}$ updates the weight of the edge $(v, u)$ in $H$ to $w_H (v, u) = \min (w_G (v, u), \delta (v, u), \delta(u, v))$, otherwise $\mathcal{C}$ adds the edge $(v, u)$
to $H$ with weight equal to $w_H (v, u) = \delta (v, u)$.
Note that the number of these changes at time $t$ of $\mathcal{A}$ is equal to
$\chi^{(t)}(v)$ for a vertex $v \in V$. Observe also that based on the two properties, the number of vertices that 
participate in $H$ during a phase of length $\ell$ is at most $\ell \cdot (2 + 2\mu)$. 
Thus we can always find a ``fresh'' vertex (of degree $0$) in $H$.
\end{enumerate}
Finally, all the changes performed to $H$ are processed by the fully dynamic distance oracle $\mathcal{B}$ running on $H$, where edge weight 
changes are simulated by a deletion followed by a re-insertion.

Now a query for the approximate distance between any pair of vertices $s$ and $t$ is answered by returning:
\begin{equation*}
\delta_{\mathcal{C}} (s, t) = \min \left( \min_{p \in S(s) \cap V(H), q \in S(t) \cap V(H)} \left( \delta (s, p) + \delta_\mathcal{B} (p, q) + \delta (t, q) \right), \delta_{\mathcal{A}} (s, t) \right). \, 
\end{equation*}
Whenever $S(s) \cap V(H) = \emptyset$ or
$S(t) \cap V(H) = \emptyset$, we let the inside term $\min(\cdot) = \infty$.

\paragraph{Correctness.}
To prove the correctness of this algorithm, we need to show that $ \dist_G (s, t) \leq \delta_{\mathcal{C}}(s, t) \leq \alpha \beta \cdot \dist_G (s, t) $.
The lower bound $ \dist_G (s, t) \leq \delta_{\mathcal{C}}(s, t)$ is immediate, since for each approximate distance returned by $\mathcal{C}$, the corresponding path uses edges from $G$ or distance estimates from the decremental approximate hub-labeling scheme which are never an underestimation of the real distance.
To prove the upper bound, consider a shortest path $\pi$ from $s$ to $t$ in $G$, and let $G_\mathcal{A}$ be the graph maintained by $\mathcal{A}$ (i.e., the edge set of $G_\mathcal{A}$ is $E(G_\mathcal{A}) = F \setminus D$). If the path $\pi$ contains only edges from the set $F \setminus D$, then
$\delta_{\mathcal{C}}(s, t) \leq \delta_{\mathcal{A}}(s, t) \leq \alpha \cdot \dist_{G_\mathcal{A}}(s, t) = \alpha \cdot \dist_G (s, t)$, and the claim follows. Otherwise, let $ (u_1, v_1), \ldots, (u_j, v_j) \in I$ denote the edges of $ \pi $ that have been inserted since the beginning of the current phase in order of appearance on $ \pi $.
Furthermore, let $ p_0 \in S(s) \cap S(u_1) $ be the vertex that ``certifies'' $ \delta_{\mathcal{A}} (s, u_1) $, that is, $ \delta_{\mathcal{A}} (s,  u_1) = \delta (s, p_0) + \delta (u_1, p_0) $. 
Similarly, let $p_j \in S(v_j) \cap S(t) $ be the vertex that ``certifies'' $ \delta_{\mathcal{A}} (v_j, t) $, and for every $ 1 \leq i \leq j-1 $, let $p_i \in S(v_i) \cap S(u_{i+1}) $ be the vertex that ``certifies'' $ \delta_{\mathcal{A}} (v_i, u_{i+1}) $. These vertices must exist by the definition of an approximate hub-labeling scheme. Furthermore,
by the construction of $H$, the edges $(p_0, u_1)$ and $(p_j, t)$ have been inserted to $H$, because $u_1 \in U$ and $p_0 \in S(u_1)$, and $v_j \in U$ and $p_j \in S(v_j)$ respectively. Hence, the vertices $p_0$ and $p_j$ belong to $V(H)$,
and the sum $\delta (s, p_0) + \delta_{\mathcal{B}} (p_0, p_j) + \delta (t, p_j)$ participates in the inside term $\min(\cdot)$. Therefore to analyze the claimed upper-bound on the stretch, we proceed as follows:
\begin{align*}
\delta_{\mathcal{C}} (s, t) &\leq \delta (s, p_0) + \delta_{\mathcal{B}} (p_0, p_j) + \delta (t, p_j) \\
 & {\scriptsize \text{(stretch guarantee of $ \mathcal{B} $)}} \\
 &\leq \delta (s, p_0) + \beta \cdot \dist_H (p_0, p_j) + \delta(t, p_j) \\
 & {\scriptsize \text{(triangle inequality)}} \\
 &\leq \delta (s, p_0) + \beta \cdot \dist_H (p_0, u_1) + \sum_{1 \leq i \leq j-1} \beta \cdot \left( \dist_H (u_i, v_i) + \dist_H (v_i, p_i) + \dist_H (p_i, u_{i+1}) \right) \\
 &~~~~+ \beta \cdot \left( \dist_H (u_j, v_j) + \dist_H (v_j, p_j) \right) + \delta(t, p_j) \\
  & {\scriptsize \text{($ \dist_H \leq w_H $)}} \\
 &\leq \delta (s, p_0) + \beta \cdot w_H (p_0, u_1) + \sum_{1 \leq i \leq j-1} \beta \cdot \left( w_H (u_i, v_i) + w_H (v_i, p_i) + w_H (p_i, u_{i+1}) \right) \\
 &~~~~+ \beta \cdot \left( w_H (u_j, v_j) + w_H (v_j, p_j) \right) + \delta(t, p_j)
 \end{align*}
 By the construction of $H$, the edges $(u_i, v_i)$ of $\pi$ and the corresponding edges $(p_{i-1}, u_i)$ and $(v_i, p_i)$ have been inserted to $H$, because $(u_i, v_i) \in I$, $u_i \in U$ and $p_{i-1} \in S(u_i)$, and $v_i \in U$ and $p_i \in S(v_i)$ respectively. Hence by the definition
of $w_H(\cdot)$, we can replace $w_H(u_i, v_i)$ with $w_G(u_i, v_i)$, $w_H(p_{i-1}, u_i)$ with $\delta(u_i,
p_{i-1})$ and $w_H(v_i, p_i)$ with $\delta(v_i, p_i)$. As a result, we have that:

 \begin{align*}
 \delta_{\mathcal{C}} (s, t) &\leq \delta (s, p_0) + \beta \cdot \delta(u_1, p_0) + \sum_{1 \leq i \leq j-1} \beta \cdot \left( w_G (u_i, v_i) + \delta (v_i, p_i) +\delta (u_{i+1}, p_i) \right) \label{ineq:constr_H} \\
 &~~~~+ \beta \cdot \left( w_G (u_j, v_j) + \delta (v_j, p_j) \right) + \delta(t, p_j) \\
 & {\scriptsize \text{($ \pi $ is a shortest path)}} \\
 &= \delta (s, p_0) + \beta \cdot \delta (u_1, p_0) + \sum_{1 \leq i \leq j-1} \beta \cdot \left( \dist_G (u_i, v_i) + \delta (v_i, p_i) + \delta (u_{i+1}, p_i) \right) \\
 &~~~~+ \beta \cdot \left( \dist_G (u_j, v_j) + \delta  (v_j, p_j) \right) + \delta(t, p_j) \\
 & {\scriptsize \text{($\beta \geq 1$)}}  \\
 &\leq \beta \cdot (\delta (s, p_0) + \delta(u_1, p_0)) + \sum_{1 \leq i \leq j-1} \beta \cdot \left(\dist_G (u_i, v_i) + \delta (v_i, p_i) + \delta (u_{i+1}, p_i) \right)  \\
 &~~~~+ \beta \cdot (\dist_G (u_j, v_j) + \delta (v_j, p_j) + \delta(t, p_j)) \\
 & {\scriptsize \text{(definition of approximate hub-labeling scheme)}} \\
 &= \beta \cdot \delta_\mathcal{A} (s, u_1) + \sum_{1 \leq i \leq j-1} \beta \cdot \left( \dist_G (u_i, v_i) + \delta_\mathcal{A} (v_i, u_{i+1}) \right) \\
 &~~~~+ \beta \cdot \left( \dist_G (u_j, v_j) + \delta_\mathcal{A} (v_j, t) \right)
 \end{align*}
 From the stretch guarantee of $\mathcal{A}$, it holds that
$\delta_\mathcal{A}(u, v) \leq \alpha \cdot d_{G_\mathcal{A}}(u, v)$ for any pair of vertices $u, v \in V$. For two vertices $v_i, u_{i+1}$ from the previous sum, we have that the subpath of $\pi$ from $v_i$ to $u_{i+1}$ uses edges only from the set $F \setminus D$, implying that $d_{G_\mathcal{A}}(v_i, u_{i+1}) = d_G(v_i, u_{i+1})$. The same argument holds for the pairs $s, u_1$ and $v_j, t$, thus it follows that:
 \begin{align*}
 \delta_{\mathcal{C}} (s, t) &\leq \alpha \beta \cdot \dist_G (s, u_1) + \sum_{1 \leq i \leq j-1} \beta \cdot \left( \dist_G (u_i, v_i) + \alpha \cdot \dist_G (v_i, u_{i+1}) \right)  \\
 &~~~~+ \beta \cdot \left( \dist_G (u_j, v_j) + \alpha \cdot \dist_G (v_j, t) \right)  \\
 & {\scriptsize \text{($ \alpha \geq 1 $)}}  \\
 &\leq \alpha \beta \cdot \dist_G (s, u_1) + \sum_{1 \leq i \leq j-1} \alpha \beta \cdot \left( \dist_G (u_i, v_i) + \dist_G (v_i, u_{i+1}) \right) \\
 &~~~~+ \alpha \beta \cdot \left( \dist_G (u_j, v_j) + \dist_G (v_j, t) \right) \\
 &= \alpha \beta \cdot \dist_G (s, t).
\end{align*}

\paragraph{Update and Query time.}
To analyze the running times, consider a fixed phase of length $\ell$. 
During the first phase, the query time is $Q_\mathcal{B}(2\ell, \ell, W)$ and the amortized update is $t_\mathcal{B}(2\ell, \ell, W)$, as the initially empty graph $G$ has at
most $2\ell$ vertices and $\ell$ edges after $\ell$ updates. For the subsequent phases we proceed as follows. By the construction of $H$ and the two properties, the graph $H$ has at most $\min(\ell(2 + 2\mu), n)$ vertices and $\ell (1 + 2\mu)$ edges during the phase, and the maximum edge weight in $H$ is $nW$ (the maximum distance in $G$).\footnote{We can assume that
$\delta(\cdot, \cdot)$ is upper bounded by $nW$ whenever it has a finite value, since the maximum
distance in $G$ is at most $nW$. Likewise, we can use the value $nW + 1$ instead of $\infty$.}
Moreover by the first property we have that $|S(s) \cap V(H)| \leq \gamma$ and $|S(t) \cap V(H)| \leq \gamma$. Therefore the query time is equal to: \begin{equation*}
Q_\mathcal{C} (n, m, W) = Q_\mathcal{A} (n, m, W) + \gamma^2 \cdot Q_\mathcal{B} (\min(\ell(2 + 2\mu), n), \ell(1 + 2\mu), nW)
\end{equation*}

Let us now analyze the amortized update time. Since the total update time of $\mathcal{A}$ is $T_\mathcal{A} (n, m, W)$ and the amortized update time of $\mathcal{B}$ is $t_\mathcal{B} (\min(\ell(2 + 2\mu), n), \ell(1 + 2\mu), nW)$
during the phase, it remains to bound the total number of updates to $H$ per phase.
Whenever an edge $e = (u, v)$ is inserted to $G$, we add to $H$ the two endpoints $u$ and $v$ together with their hub sets $S(u)$ and $S(v)$, and at most $1 + 2\gamma$ updates can occur to $H$. Until $(u, v)$ is deleted from $H$, every update to $H$ between vertices 
from the set $\{u\} \cup \{v\} \cup S(u) \cup S(v)$ modifies an entry of the distance estimate $\delta(u, \cdot)$ or $\delta(v, \cdot)$. 
By the definition of $\chi^{(t)}(\cdot)$, the number of times that the distance estimates $\delta(u, \cdot)$ and $\delta(v, \cdot)$ are modified at time $t$ of $\mathcal{A}$, is equal to $\chi^{(t)}(u) + \chi^{(t)}(v)$. Hence until $(u, v)$ is deleted from $H$, the total number of updates to $H$ between vertices from the set $\{u\} \cup \{v\} \cup S(u) \cup S(v)$ is equal to $2 \cdot (\sum_{t} \chi^{(t)}(u) + \sum_{t} \chi^{(t)}(v)) = 2 \cdot (X(u) + X(v))$,\footnote{We multiply by $2$ because edge weight changes are simulated by a deletion followed by a re-insertion.} which is at most $4\zeta$ based on the second property of Lemma~\ref{lem:decr_to_fully_dynamic_dist_oracle}. Moreover, when the edge $e$ is deleted from $G$, at most $1 + 2\gamma$ updates can occur to $H$.
Therefore, the total number of updates to $H$ that correspond to an inserted edge in $G$, is at most $2 + 4\gamma + 4\zeta = 2 + 4\mu$ per phase. Since there can be at most $\ell$ inserted edges per phase, the total number of updates to $H$ during a phase is at most $\ell (2 + 4\mu)$. This implies that the total time for processing all updates is $T_\mathcal{A} (n, m, W) + t_\mathcal{B} (\min(\ell(2 + 2\mu), n), \ell(1 + 2\mu), 
nW) \cdot \ell (2 + 4\mu)$, which (when amortized over the $\ell$ updates of the previous phase) amounts to an amortized update time of:
\begin{equation*}
T_{\mathcal{C}} (n, m, W) = \frac{T_\mathcal{A} (n, m, W)}{\ell} + t_\mathcal{B} (\min(\ell(2 + 2\mu), n), \ell(1 + 2\mu), nW) \cdot (2 + 4\mu)
\end{equation*}
\end{proof}
\subsection{Decremental approximate hub-labeling scheme}\label{subsec:decremental}
In this section, we argue that an existing decremental distance oracle from \cite{LackiN22} also provides an approximate hub-labeling scheme whose properties make the the reduction of Lemma~\ref{lem:decr_to_fully_dynamic_dist_oracle} quite efficient. This decremental
algorithm is based on the well-known static Thorup-Zwick (TZ)  distance oracle \cite{TZ2005}.

  \paragraph{Thorup-Zwick distance oracle.} Given a graph $G = (V, E)$, the construction starts by defining a non-increasing sequence of sets $V=A_0 \supseteq A_1 \supseteq \cdots \supseteq A_k=\emptyset$, where for each $1 \leq i < k$, the set $A_i$ is obtained by subsampling each element of $A_{i-1}$ independently with probability $n^{-1/k}$.

For every vertex $v \in V$ and $1 \leq i < k$, let $\delta(v, A_i) = \min_{u \in A_i} \dist_G(v, u)$ be the minimum
distance from $v$ to a vertex in $A_i$.
As $A_k = \emptyset$, we let $\delta(v, A_k) = \infty$. Moreover, let $p_i(v) \in A_i$ be a vertex in $A_i$ closest to $v$, that is, $\dist_G(v, p_i(v)) = \delta(v, A_i)$. Then, the bunch $B(v) \subseteq V$ of each $v \in V$ is defined as:
\begin{align*}
    B(v) = \bigcup_{i=0}^{k-1} B_i(v) \;\text{,~~where~~} B_i(v) = \{u \in A_i \setminus A_{i+1}: \dist_G(v, u) < \dist_G(v, A_{i+1})\}
\end{align*}
The cluster of a vertex $u \in A_i \setminus A_{i+1}$ is defined as $C(u) = \{v \in V: \dist_G(v, u) < \dist_G(v, A_{i+1})\}$. Observe that $u \in B(v)$ if and only if $v \in C(u)$, for any $u, v \in V$. 

As noted in \cite{TZ2005}, this construction is a hub-labeling scheme of stretch $2k - 1$ (see Definition \ref{def:apprx_hub_lab_schm}), where the hub set $S(v)$ of a vertex $v \in V$ is $S(v) = B(v) \cup (\bigcup_{i=0}^{k-1} \{p_i(v)\} )$. In other words, bunches and pivots of all the $k$ levels form a hub set for $v$. For obtaining the distance estimates $\delta(v, \cdot)$ 
for all $v \in V$ as in Definition \ref{def:apprx_hub_lab_schm}, we need the associated distances $\delta(v, u) = \dist_G(v, u)$ for all $u \in S(v)$.
It can be shown that with a simple modification of the stretch argument (e.g.~see \cite{HKN14}), it is enough to only use the bunches as the hub sets, and explicit access to pivots is not necessary. Hence for simplifying the presentation in this section
we assume that the hub sets are equivalent with the bunches.
As shown in \cite{TZ2005}, the size of the bunch of any vertex is w.h.p.~bounded by $\tilde{O}(n^{1/k})$.
Recall that the maximum hub set size is one of the parameters governing the efficiency of our reduction.

In the following, we review the decremental algorithm of \cite{RZ12} which maintains TZ distance oracles for $d$-bounded distances, and the decremental algorithm of \cite{LackiN22} which has good properties
for the reduction of Lemma~\ref{lem:decr_to_fully_dynamic_dist_oracle}.

\paragraph{Decremental algorithm for approximate TZ distance oracle.}
We use the decremental algorithm by \cite{LackiN22} that satisfies the properties of Lemma \ref{lem:decr_to_fully_dynamic_dist_oracle},
for sufficiently small $\gamma$ and $\zeta$.
The properties that we need are implicit in their analysis. We rephrase their guarantees and give a high-level proof sketch for completeness, but we refer the reader to \cite{LackiN22} for further details.

The algorithm of \cite{LackiN22} utilizes a decremental version of the TZ distance oracles for $d$-bounded distances by \cite{RZ12} on a sequence of graphs. A crucial property of \cite{RZ12} algorithm can be summarized in the following lemma.

 \begin{lemma}[Implicit in \cite{RZ12}] \label{lem:RZ_bound}
For every vertex $v \in V$ and $0 \leq i < k$, there is a decremental algorithm that maintains the bunches and the estimates $\delta$ up to a distance bound $d$. Over the sequence of updates, the expected number of times $\delta(v, u)$ changes for all vertices $u \in A_i \setminus A_{i+1}$ such that $v \in C(u)$ and $\dist_G(v, u) \leq d$ is $\tilde{O}(dn^{1/k})$. Equivalently, w.h.p~the number of times $B(v)$ or a corresponding distance estimate $\delta(v, u)$ for $u \in B(v)$ changes over all updates is bounded by $\tilde{O}(dn^{1/k})$.

\end{lemma}

This lemma allows us to bound the number of changes in the bunches (as required by Lemma \ref{lem:decr_to_fully_dynamic_dist_oracle}) for pairs of vertices that are within bounded distances up to $d$. In terms of $X(v)$ as defined in Section~\ref{subsec:reduction} this implies $X(v) \leq \tilde{O}(dn^{1/k})$. 

 This algorithm is not efficient when we are not restricted to bounded distance. In order to eliminate this dependence on $d$, in \cite{LackiN22} they use decremental hopsets with a hopbound $\beta$ that informally speaking, allow them to do the following. Instead of working on the original graph, they maintain the decremental distance oracle of \cite{RZ12} on a sequence of scaled graphs up to depth $\beta$ in time $\tilde{O}(\beta mn^{1/k})$. A $(\beta, 1+\epsilon)$-hopset $H'$ for $G=(V,E)$ is a set of weighted edges such that for all $u,v \in V$ we have $\dist_G(u,v) \leq \dist^{(\beta)}_{G\cup H'}(u,v)  \leq (1+ \epsilon)\dist_G(u,v)$, where $\dist^{(\beta)}_{G\cup H'}(u,v) $ refers to a shortest path that uses at most $\beta$ hops. A decremental hopset has the additional property that the edge weights are non-decreasing. By maintaining a hopset with hopbound $\beta= \textit{polylog } (n)$ they can maintain TZ distance oracles with the following guarantees:

\begin{lemma}[Implicit in \cite{LackiN22}] \label{lem:dec_alg}
    Given a weighted undirected graph $G = (V,E)$ and $k > 1, 0 < \epsilon < 1$, there is a decremental hub-labeling scheme of stretch $(2k - 1)(1 + \epsilon)$ and w.h.p. the total update time is $\tilde{O}(mn^{1/k}) \cdot O(\log nW/\epsilon)^{2k+1}$. Moreover w.h.p. we have the following two properties:
    \begin{enumerate}
        \item $\forall v \in V$ and $\forall t: |S^{(t)}(v)| \leq \tilde{O}{(n^{1/k})}$. In other words, the size of the bunch of any vertex is bounded by $\tilde{O}{(n^{1/k})}$ at any moment of the algorithm.
        \item $\forall v \in V: X(v) \leq \tilde{O}{(n^{1/k})}\cdot O(\log nW/\epsilon)^{2k+1}$. In other words, for every vertex $v \in V$ the total number of changes to $\delta(v, \cdot)$ is at most $\tilde{O}{(n^{1/k})\cdot O(\log nW/\epsilon)^{2k+1}}$ 
        over the course of the algorithm. Moreover the algorithm detects and reports these changes explicitly.
    \end{enumerate}
\end{lemma} 

\begin{proof}[Proof sketch.] 
    In \cite{LackiN22} approximate TZ bunches and clusters (and hence hub sets) are maintained, roughly as follows: a decremental $(\beta, 1+\epsilon)$- hopset $H'$ with hopbound $\beta = O(\log nW/\epsilon)^{2k+1}$ and size $O(n^{1+1/k})$ can be maintained in $\tilde{O}(\beta mn^{1/k})$ total update time. Then $H'$ is used to maintain a $(2k-1)(1+\epsilon)$ distance oracle by running the Roditty-Zwick \cite{RZ12} algorithm on a sequence of $O(\log nW)$ scaled graphs $G_1,\ldots, G_{\log nW}$ up to depth $\beta$.  Roughly speaking, this scaling approach (originally proposed by \cite{KS97}) for a fixed error parameter $\epsilon_0$ and hopbound parameter $\beta$ maintains a graph $G_r$ 
    for each distance interval $[2^r,2^{r+1}]$ in which the edge weights are rounded such that $\beta$-\textit{hop-bounded} distances of length $\in [2^r,2^{r+1}]$ in the original graph $G$ can be $(1+\epsilon_0)$-approximated by $O(\beta/\epsilon_0)$-\textit{bounded depth} distances on the scaled graph $G_r$. Hence by first adding the hopset edges $H'$ and then applying the scaling to $G \cup H'$, it is enough to consider $\beta$ bounded distances on each scaled graph. 
    The final bunch of each vertex is the union of bunches over all the graphs $G_i$, for $1 \leq i \leq \log nW$. Since the size of the bunches on each scaled graph at any time is also $\tilde{O}(n^{1/k})$, the first property holds. 
    
    Also from Lemma \ref{lem:RZ_bound} and by setting $d= O(\beta/\epsilon_0) = \text{polylog} (n)$ and we have:
    \begin{itemize}
        \item  The total number of times the bunch $B(v)$ changes for each vertex $v$ on \textit{each scaled graph} is~w.h.p~$\tilde{O}(\beta n^{1/k})$. Hence the second property holds since we have $\beta = O(\log (nW)/\epsilon)^{2k+1}$, and in total $v$ changes its bunch on union of $O(\log nW)$ scaled graphs at most $\tilde{O}(\beta n^{1/k})$ times. 
 
    \item The update time is $\tilde{O}(mn^{1/k}\beta)$.
    Analyzing correctness of their hierarchical decremental hopsets requires handling some technicalities that we do not get into here as we can use their stretch analysis as a black-box. 
    \end{itemize}
      
    Overall the distance oracle based on the \textit{approximate} bunches maintained lead to $(2k-1) (1+\epsilon)$-approximate distances. Similar to the static case, the stretch guarantee of \cite{TZ2005} carries over to the hub-labeling scheme that are based on approximate bunches as the hub set. In this case in addition to the $(2k-1)$ factor, there will be an additional $(1+\epsilon)$ factor since the scaling and use of hopsets effectively give us approximate bunches.
 \end{proof}
 
\subsection{Fully dynamic distance oracle}\label{subse:final_result}

In this section we explain how we obtain our final fully dynamic distance oracle by using the decremental algorithm of Section \ref{subsec:decremental} in our reduction of Lemma \ref{lem:decr_to_fully_dynamic_dist_oracle}.

\begin{theorem} \label{thm:general_tradeoffs}
For any integer parameters $i \geq 0, k>1$, there is a fully dynamic distance oracle $ \mathcal{B}_i $ with stretch $(4k)^i$ and w.h.p.~the amortized update time is $ t_{\mathcal{B}_i} (n, m, W) = \tilde{O}(1)^{ki} \cdot m^{3/(3i+1)} \cdot n^{4i / k}$ and the query time $Q_{\mathcal{B}_i} (n, m, W) = \tilde{O}(1)^i \cdot n^{2i / k}$.
\end{theorem}
\begin{proof}
The proof is by induction on $i$. For the base case $i = 0$, let $\mathcal{B}_0$ be the trivial fully dynamic distance oracle 
that achieves stretch~$1$, amortized update time $ t_{\mathcal{B}_0} (n, m, W) = O (n^3) $, and query time $Q_{\mathcal{B}_0} (n, m, W) = O(1) $, by recomputing all-pairs shortest paths from scratch after each update (e.g., with the Floyd–Warshall algorithm).

For the induction step, let $\mathcal{A}$ denote the decremental approximate hub-labeling scheme from Lemma~\ref{lem:dec_alg} with stretch $\alpha = 4k$ and w.h.p. total update time $T_{\mathcal{A}} (n, m, W) = \tilde{O}(1)^k \cdot mn^{1/k}$ and query time $Q_{\mathcal{A}}(n, m, W) = 
\tilde{O}(1) \cdot n^{1/k}$, where $\epsilon$ has been replaced with any value strictly smaller than $\frac{1}{2}$.
By inductive hypothesis, we have that $ \mathcal{B}_i $ (with $ i \geq 0 $) is a fully dynamic distance oracle of stretch $\beta_i = (4k)^i$ with amortized update time 
$\tilde{O}(1)^{ki} \cdot m^{3/(3i+1)} \cdot n^{4i / k}$ and query time $\tilde{O}(1)^i \cdot n^{2i / k}$.
Based on Lemma~\ref{lem:dec_alg}, the decremental approximate hub-labeling scheme $\mathcal{A}$ satisfies the properties of Lemma~\ref{lem:decr_to_fully_dynamic_dist_oracle} with $\gamma = \tilde{O}(1) \cdot n^{1/k}$ and $\zeta = \tilde{O}(1)^k \cdot n^{1/k}$. By applying then Lemma~\ref{lem:decr_to_fully_dynamic_dist_oracle} to $\mathcal{A}$ and 
$\mathcal{B}_i $ with $\ell = m^{(3i + 1) / (3i + 4)}$, the resulting fully dynamic distance oracle $ \mathcal{B}_{i+1} $ has stretch $(4k)^{i+1}$, and amortized update time:
\begin{equation*}
t_{\mathcal{B}_{i+1}}(n, m, W) = \frac{T_\mathcal{A} (n, m, W)}{\ell} + t_{\mathcal{B}_i} (n, \ell(1 + 2\mu), nW) \cdot (2 + 4\mu)
\end{equation*}
The first term is equal to:
\begin{align*}
\frac{T_\mathcal{A} (n, m, W)}{\ell} &= \frac{\tilde{O}(1)^k \cdot mn^{1/k}}{\ell} \\
 & {\scriptsize \text{(Replace $\ell$ with $m^{(3i + 1) / (3i + 4)}$)}} \\
 &=\tilde{O}(1)^k \cdot m^{3/(3i+4)} \cdot n^{1/k} \\
 &=\tilde{O}(1)^k \cdot m^{3/(3(i+1)+1)} \cdot n^{1/k}
\end{align*}
and the second term is equal to (where $\mu = \tilde{O}(1)^k \cdot n^{1/k}$): 
\begin{align*}
t_{\mathcal{B}_i} (n, \ell(1 + 2\mu), nW) \cdot (2 + 4\mu) &= \tilde{O}(1)^{ki} \cdot (\ell \cdot \tilde{O}(1)^k \cdot n^{1/k})^{3/(3i+1)} \cdot n^{4i / k} \cdot \tilde{O}(1)^k \cdot n^{1/k} \\
 & {\scriptsize \text{(Replace $\ell$ with $m^{(3i + 1) / (3i + 4)}$)}} \\
 &=\tilde{O}(1)^{ki} \cdot (m^{(3i + 1) / (3i + 4)} \cdot \tilde{O}(1)^k \cdot n^{1/k})^{3/(3i+1)} \cdot n^{4i / k} \cdot \tilde{O}(1)^k \cdot n^{1/k}\\
  & {\scriptsize \text{(Replace $n^{3/(3i+1)k}$ with $n^{3/k}$ and $\tilde{O}(1)^{3k/(3i+1)}$ with $\tilde{O}(1)^{3k}$)}} \\
  &=\tilde{O}(1)^{ki} \cdot m^{3 / (3i + 4)} \cdot \tilde{O}(1)^{3k} \cdot n^{3/k} \cdot n^{4i / k} \cdot \tilde{O}(1)^k \cdot n^{1/k} \\
 &=\tilde{O}(1)^{ki+k} \cdot m^{3 / (3i + 4)} \cdot n^{(4i + 4) / k} \\
 &=\tilde{O}(1)^{k(i+1)} \cdot m^{3 / (3(i+1) + 1)} \cdot n^{4(i+1) / k}
\end{align*}
Therefore the amortized update time of $\mathcal{B}_{i+1}$ is:
\begin{equation*}
t_{\mathcal{B}_{i+1}}(n, m, W) = \tilde{O}(1)^{k(i+1)} \cdot m^{3 / (3(i+1) + 1)} \cdot n^{4(i+1) / k}
\end{equation*}
Finally the query time of $\mathcal{B}_{i+1}$ is (where $\gamma^2 = \tilde{O}(1)^2 \cdot n^{2/k}$): 
\begin{align*}
    Q_{\mathcal{B}_{i+1}} (n, m, W) &= Q_\mathcal{A} (n, m, W) + \gamma^2 \cdot Q_{\mathcal{B}_i} (n, \ell(1 + 2\mu), nW) \\ 
    &= \tilde{O}(1) \cdot n^{1/k} + \tilde{O}(1)^2 \cdot n^{2/k} \cdot \tilde{O}(1)^i \cdot n^{2i/k} \\
    &= \tilde{O}(1)^{i+1} \cdot n^{2(i+1)/k}
\end{align*}
and so the distance oracle $\mathcal{B}_{i+1}$ has the desired guarantees.
\end{proof}

\maincor*
\begin{proof}
By Theorem~\ref{thm:general_tradeoffs}, for any $i \geq 1, k > 1$, there is a fully dynamic distance oracle $\mathcal{B}_i$ of
stretch $(4k)^i$ that w.h.p. achieves
$\tilde{O}(1)^{ki} \cdot m^{1/i} \cdot n^{4i / k}$
amortized update time and $\tilde{O}(1)^i \cdot n^{2i/k}$
query time. By setting then $i = \frac{4}{\rho}$ and $k = \frac{64}{\rho^2}$, the claim follows.
\end{proof}

In Theorem~\ref{thm:general_tradeoffs}, we can set $i$ to be a constant and set $k= O(\log \log n)^{1/i}$ to obtain another tradeoff, which is summarized in the following corollary.

\begin{corollary}\label{cor:loglog_stretch}
Given a weighted undirected graph $G=(V,E)$, there is a fully dynamic distance oracle with stretch $O(\log \log n)$ that w.h.p. achieves $n^{o(1)}$ query time and
$\tilde{O}(n^{\rho})$ amortized update time, for an arbitrarily small constant $\rho$.
\end{corollary}

Finally note that we can also obtain similar tradeoffs as \cite{ForsterGH21} where all three of stretch, amortized update time and query time are $n^{o(1)}$, by setting $k= O(\log \log n)^2$ and $i= O(\log \log n)$ in Theorem \ref{thm:general_tradeoffs}.



\printbibliography 

\end{document}